\documentclass[english]{article}

\usepackage[T1]{fontenc}
\usepackage[latin1]{inputenc}
\usepackage{geometry}
\usepackage{float}
\usepackage{mathrsfs}
\usepackage{amsmath}
\usepackage{amsthm}
\usepackage{amsfonts,mathtools,bm}
\usepackage{amssymb}
\usepackage{graphicx}
\usepackage{subfigure}
\usepackage{algorithm}
\usepackage{algpseudocode}
\usepackage{babel}

\makeatletter

\newcounter{hypA}
\newenvironment{hypA}{\refstepcounter{hypA}\begin{itemize}
  \item[({\bf A\arabic{hypA}})]}{\end{itemize}}
\newcounter{hypB}

\floatstyle{ruled}
\newfloat{algorithm}{tbp}{loa}
\providecommand{\algorithmname}{Algorithm}
\floatname{algorithm}{\protect\algorithmname}

\makeatother

\DeclareMathOperator{\erf}{erf}
\DeclareMathOperator*{\argmin}{argmin}

\newtheorem{theorem}{Theorem}[section]
\newtheorem{lem}{Lemma}[section]

\def\tr{\mathrm{t}}  
\def\bbR{\mathbb{R}}
\def\calN{\mathcal{N}}

\date{}

\begin{document}

\begin{center}

{\Large \textbf{On Large Lag Smoothing for Hidden Markov Models}}

\vspace{0.5cm}

BY JEREMIE HOUSSINEAU$^{1}$, AJAY JASRA$^{1}$ \& SUMEETPAL SINGH$^{2}$

{\footnotesize $^{1}$Department of Statistics \& Applied Probability,
National University of Singapore, Singapore, 117546, SG.}
{\footnotesize E-Mail:\,}\texttt{\emph{\footnotesize stahje@nus.edu.sg, staja@nus.edu.sg}}\\
{\footnotesize $^{2}$Department of Engineering,
University of Cambridge, Cambridge, CB2 1PZ and The Alan Turing Institute, UK.}
{\footnotesize E-Mail:\,}\texttt{\emph{\footnotesize sss40@cam.ac.uk}}
\end{center}

\begin{abstract}
\noindent In this article we consider the smoothing problem for hidden Markov models (HMM).
Given a hidden Markov chain $\{X_n\}_{n\geq 0}$ and observations $\{Y_n\}_{n\geq 0}$, our
objective is to compute $\mathbb{E}[\varphi(X_0,\dots,X_k)|y_{0},\dots,y_n]$ for some real-valued, integrable 
functional $\varphi$ and $k$ fixed, $k \ll n$ and for some realisation $(y_0,\dots,y_n)$ of $(Y_0,\dots,Y_n)$. We introduce a novel application of the multilevel Monte Carlo (MLMC) method with a coupling
based on the Knothe-Rosenblatt rearrangement. We prove that this method can approximate the afore-mentioned quantity with a mean square error (MSE) of $\mathcal{O}(\epsilon^2)$,
for arbitrary $\epsilon>0$ with a cost of $\mathcal{O}(\epsilon^{-2})$. This is in contrast to the same direct Monte Carlo method, which requires a cost of $\mathcal{O}(n\epsilon^{-2})$ for the same MSE.
The approach we suggest is, in general, not possible to implement, so the optimal transport methodology of \cite{span} is used, which directly
approximates our strategy. We show that our theoretical improvements are achieved, even under approximation, in several numerical examples. 
\end{abstract}

\section{Introduction}

Given a hidden Markov chain $\{X_n\}_{n\geq 0}$, $X_n\in\mathsf{X}\subset\mathbb{R}^d$ and observations $\{Y_n\}_{n\geq 0}$, $Y_n\in\mathsf{Y}$, we consider a probabilistic model such
that for Borel $A\in\mathcal{X}$, $\mathbb{P}(X_0\in A) = \int_A f(x) dx$, for every $n\geq 1$, $x_{0:n-1}\in\mathsf{X}^n$
\begin{equation}
\label{eq:prediction}
\mathbb{P}(X_n\in A|x_{0:n-1}) = \int_A f(x_{n-1},x) dx 
\end{equation}
with $dx$ Lebesgue measure and for Borel $B\in\mathcal{Y}$ and all $n\geq 0$, $(y_{0:n-1},x_{0:n})\in\mathsf{Y}^n\times\mathsf{X}^{n+1}$
\begin{equation}
\label{eq:observation}
\mathbb{P}(Y_n\in B|y_{0:n-1},x_{0:n}) = \int_B g(x_n,y) dy,
\end{equation}
where we have used the compact notation $a_{k:n}=(a_k,\dots,a_n)$ for any $k,n \geq 0$ and any sequence $(a_n)_{n\geq 0}$ with the convention that the resulting vector of objects is null if $k>n$.
The model defined by \eqref{eq:prediction} and \eqref{eq:observation} is termed a hidden Markov model. In this article, given $y_{0:n}$, our objective is to compute $\mathbb{E}[\varphi(X_{0:k})|y_{0:n}]$ for some real-valued, integrable 
functional $\varphi$ and $k$ fixed, $k \ll n$, which we refer to as large-lag smoothing. Hidden Markov models and the smoothing problem are found in many real applications, such as finance, genetics and engineering; see e.g.~\cite{Cappe_2005} and the references therein.

The smoothing problem is notoriously challenging. Firstly, $\mathbb{E}[\varphi(X_{0:k})|y_{0:n}]$ is seldom available analytically and hence numerical methods are required.  Secondly, if one wants to compute $\mathbb{E}[\varphi(X_{0:k})|y_{0:n}]$
for several values of $n$, i.e.~potentially recursively, then several of the well-known methods for approximation of $\mathbb{E}[\varphi(X_{0:k})|y_{0:n}]$ can fail. For instance the particle filter (e.g.~\cite{doucet_johan} and the references therein) suffers from the well-known path degeneracy
problem (see e.g.~\cite{kantas}). Despite this, several methods are available for the approximation of $\mathbb{E}[\varphi(X_{0:k})|y_{0:n}]$, such as particle Markov chain Monte Carlo \cite{andrieu} or the PaRIS algorithm \cite{olsson}, which might be considered the current state-of-the-art.
The latter algorithm relies on approximating $\mathbb{E}[\varphi(X_{0:k})|y_{0:n^*}]$ for some $n^*<n$ and is then justified on the basis of using \emph{forgetting} properties of the smoother (see e.g.~\cite{Cappe_2005,dds1}). We will extend this notion as will be explained below.

The main approach that is followed in this paper, is to utilize the multilevel Monte Carlo method (e.g.~\cite{giles,hein}). Traditional applications of this method are associated to discretizations of continuum problems, but we adopt the framework
in a slightly non-standard way. To describe the basic idea, suppose one is interested in $\mathbb{E}_{\pi}[\varphi(X)]$ for $\pi$ a probability, $\varphi$ real-valued and bounded, but, one can only hope to approximate $\mathbb{E}_{\pi_l}[\varphi(X)]$ with $\pi_l$ a probability (assumed on the same space as $\pi$), $l\in\mathbb{N}$ and in some loose sense
one has $\pi_l$ approaches $\pi$ as $l$ grows. Now, given $\pi_0,\dots,\pi_L$ a sequence of increasingly more `precise' probability distributions on the same space, one trivially has
$$
\mathbb{E}_{\pi_L}[\varphi(X)] = \mathbb{E}_{\pi_0}[\varphi(X)] + \sum_{l=1}^L\{\mathbb{E}_{\pi_{l}}[\varphi(X)]-\mathbb{E}_{\pi_{l-1}}[\varphi(X)]\}.
$$
The approach is now to sample \emph{dependent} couplings of $(\pi_l,\pi_{l-1})$ independently for $1\leq l \leq L$ and approximate the difference $\mathbb{E}_{\pi_{l}}[\varphi(X)]-\mathbb{E}_{\pi_{l-1}}[\varphi(X)]$ using Monte Carlo.
The term $\mathbb{E}_{\pi_0}[\varphi(X)]$ is also approximated using Monte Carlo with i.i.d.~sampling from $\pi_0$. Then, given a `good enough' coupling and a characterization of the bias, for many practical problems the cost to achieve
a pre-specified MSE against i.i.d.\ sampling from $\pi_L$ and Monte Carlo, is significantly reduced.

We leverage the idea of MLMC where the `level' $l$ corresponds to the time parameter and $L$ is some chosen $n^*$, so as to achieve a given level of bias. The main issue is then how to sample
from couplings which are good enough. We show that when $d=1$ (the dimension of the hidden state) that using the optimal coupling, in terms of squared Wasserstein distance, can yield
significant improvements over the case where one directly approximates $\mathbb{E}[\varphi(X_{0:k})|y_{0:n}]$ with Monte Carlo and i.i.d~sampling from the smoother. That is, for $\epsilon>0$
given, to achieve a mean square error of $\mathcal{O}(\epsilon^2)$, the cost is $\mathcal{O}(\epsilon^{-2})$, whereas for the ordinary Monte Carlo method the cost is $\mathcal{O}(n\epsilon^{-2})$.
The same conclusion with $d>1$ can be achieved using the Knothe-Rosenblatt rearrangement. The main issue with our approach is that it cannot be implemented for most problems of practical interest.
However, using the methodology in \cite{span}, it can be approximated. We show that in numerical examples our predicted theory is verified, even under this approximation. We also compare our
method directly with PaRIS, showing substantial improvement in terms of cost for a given level of MSE.

This article is structured as follows. In Section \ref{sec:approach} we detail our approach and theoretical results. In Section \ref{sec:method} we demonstrate how our approach can be implemented in practice.
In Section \ref{sec:case} we give our numerical examples. 
Section \ref{sec:summ} summarizes the article.
The appendix includes the assumptions, technical results and proofs of our main results.

\subsection{Notations}

Let $(\mathsf{X},\mathcal{X})$ be a measurable space. 
For $\varphi:\mathsf{X}\rightarrow\mathbb{R}$ we write $\mathcal{B}_b(\mathsf{X})$
and $\textrm{Lip}(\mathsf{X})$ as the collection of bounded measurable and Lipschitz functions respectively.
For $\varphi\in\mathcal{B}_b(\mathsf{X})$, we write the supremum norm $\|\varphi\|=\sup_{x\in\mathsf{X}}|\varphi(x)|$.
For $\varphi\in\mathcal{B}_b(\mathsf{X})$, $\textrm{Osc}(\varphi)=\sup_{(x,y)\in\mathsf{X}\times\mathsf{X}}|\varphi(x)-\varphi(y)|$ and 
we write $\textrm{Osc}_1(\mathsf{X})$ for the set of functions $\varphi$ on $\mathsf{X}$ such that $\textrm{Osc}(\varphi)=1$.
For $\varphi\in\textrm{Lip}(\mathsf{X})$, we write the Lipschitz constant $\|\varphi\|_{\textrm{Lip}}$.
$\mathscr{P}(\mathsf{X})$  denotes the collection of probability measures on $(\mathsf{X},\mathcal{X})$.
For a measure $\mu$ on $(\mathsf{X},\mathcal{X})$
and a $\varphi\in\mathcal{B}_b(\mathsf{X})$, the notation $\mu(\varphi)=\int_{\mathsf{X}}\varphi(x)\mu(dx)$ is used.
Let $K:\mathsf{X}\times\mathcal{X}\rightarrow[0,1]$ be a Markov kernel and $\mu$ be a measure then we use the notations
$
\mu K(dy) = \int_{\mathsf{X}}\mu(dx) K(x,dy)
$
and for $\varphi\in\mathcal{B}_b(\mathsf{X})$, 
$
K(\varphi)(x) = \int_{\mathsf{X}} \varphi(y) K(x,dy).
$
For a sequence of Markov kernels $K_1,\dots,K_n$ we write 
$$
K_{1:n}(x_0,dx_n) = \int_{\mathsf{X}^{n-1}}\prod_{p=1}^n K_p(x_{p-1},dx_p) .
$$
For $\mu,\nu\in\mathscr{P}(\mathsf{X})$, the total variation distance 
is written $\|\mu-\nu\|_{\textrm{tv}}=\sup_{A\in\mathcal{X}}|\mu(A)-\nu(A)|$.
For $A\in\mathcal{X}$ the indicator is written $\mathbb{I}_A(x)$.
$\mathcal{U}_A$ denotes the uniform distribution on the set~$A$. $\mathcal{N}(a,b)$ is the one-dimensional Gaussian distribution of mean $a$ and variance $b$.

\section{Model and Approach}\label{sec:approach}

We are given a HMM and we seek to compute
$$
\mathbb{E}_{\pi_{n,0}}[\varphi(X_0)|y_{0:n}] = \frac{\int_{\mathsf{X}^{n+1}}\varphi(x_0)\prod_{p=0}^{n} g(x_p,y_p) f(x_{p-1},x_p)dx_{0:n}}{\int_{\mathsf{X}^{n+1}}\prod_{p=0}^{n} g(x_p,y_p) f(x_{p-1},x_p)dx_{0:n}}
$$
where $f(x_{-1},x_0):=f(x_0)$  and 
for ease of simplicity we suppose that $\varphi\in\mathcal{B}_b(\mathsf{X})\cap\textrm{Lip}(\mathsf{X})$ and $\mathsf{X}$ is a compact subspace of the real line.
$\pi_{n,0}$ is the probability density (we also use the same symbol for probability measure) of the smoother given $n$ observations at the co-ordinate at time 0. That is
$$
\pi_{n,0}(x_0|y_{0:n})\propto\int_{\mathsf{X}^n}\prod_{p=0}^{n} g(x_p,y_p) f(x_{p-1},x_p)dx_{1:n}.
$$
Let $0<n^*<n$ be fixed, then we propose to consider
$$
\mathbb{E}_{\pi_{n^*,0}}[\varphi(X_0)|y_{0:n^*}] = \mathbb{E}_{\pi_{0,0}}[\varphi(X_0)|y_{0}] + \sum_{p=1}^{n^*}\{\mathbb{E}_{\pi_{p,0}}[\varphi(X_0)|y_{0:p}] - \mathbb{E}_{\pi_{p-1,0}}[\varphi(X_0)|y_{0:p-1}]\}.
$$

\subsection{Case $\mathsf{X}\subset\mathbb{R}$}\label{sec:case_R}

Let us denote the CDF of $\pi_{p,0}$ as $\Pi_{p,0}$. 
An approximation of $\mathbb{E}_{\pi_{p,0}}[\varphi(X_0)|y_{0:p}] -$ $ \mathbb{E}_{\pi_{p-1,0}}[\varphi(X_0)|y_{0:p-1}]$ is
$$
\frac{1}{N_p}\sum_{i=1}^{N_p} [\varphi(\Pi_{p,0}^{-1}(U^i))-\varphi(\Pi_{p-1,0}^{-1}(U^i))]
$$
where for $i\in\{1,\dots,N_p\}$, $U^i\stackrel{\textrm{i.i.d.}}{\sim}\mathcal{U}_{[0,1]}$ and $\Pi_{p,0}^{-1}$ is the (generalized) inverse CDF of $\Pi_{p,0}$. If we do this independently for each $p\in\{1,\dots,n\}$ 
and use an independent estimator $\frac{1}{N_0}\sum_{i=1}^N \varphi(\Pi_0^{-1}(U^i))$ for $\mathbb{E}_{\pi_{0,0}}[\varphi(X_0)|y_{0}]$
one can estimate $\mathbb{E}[\varphi(X_0)|y_{0:n}]$. The utility of the coupling is that it is optimal in terms of 2-Wasserstein distance. We have the following result, where the assumption and proof are in the appendix.

\begin{theorem}\label{theo:thm1}
Assume  (A\ref{ass:1}). Then there exists $\rho\in(0,1)$, $C<+\infty$ such that for any $\varphi\in\mathcal{B}_b(\mathsf{X})\cap\textrm{\emph{Lip}}(\mathsf{X})$, 
$n^*\geq p\geq 1$, $N_p\geq 1$, 
we have
$$
\mathbb{V}\textrm{\emph{ar}}\Big[\frac{1}{N_p}\sum_{i=1}^{N_p} [\varphi(\Pi_{p,0}^{-1}(U^i))-\varphi(\Pi_{p-1,0}^{-1}(U^i))]\Big] \leq \frac{C\rho^{p-1}\|\varphi\|_{\textrm{\emph{Lip}}}^2}{N_p}.
$$
\end{theorem}

The main implication of the result is the following. In the approach to be considered later in this paper
the cost of computing (an approximation of) $(\Pi_{p,0}^{-1},\Pi_{p-1,0}^{-1})$ is $\mathcal{O}(1)$ per time step.
So the cost of this method is $C(n^*+\sum_{p=0}^{n^*} N_p)$. Thus the MSE and cost associated to this
algorithm are (at most in the first case)
$$
C(\|\varphi\|^2\vee \|\varphi\|_{\textrm{Lip}}^2)\Big(\frac{1}{N_0}+\sum_{p=1}^{n^*}\frac{\rho^{p-1}}{N_p} + \rho^{2n}\Big) 
$$
and
\begin{equation}
\label{eq:cost_ml}
C(n^*+\sum_{p=0}^{n^*} N_p).
\end{equation}
Let $\epsilon>0$ be given. To achieve an MSE of $\mathcal{O}(\epsilon^2)$ we can choose $n^*=|\log(\epsilon)/\log(\rho)|$  (here we of course mean $n^*=\lceil |\log(\epsilon) /\log(\rho)| \rceil$, but this is omitted for simplicity)  and $N_p=\epsilon^{-2}(p+1)^{-1-\delta}$ for any $\delta>0$
yields that the associated cost is $\mathcal{O}(\epsilon^{-2})$. If one just approximates $\mathbb{E}_{\pi_{n,0}}[\varphi(X_0)|y_{0:n}]$ using 
$$
\frac{1}{N}\sum_{i=1}^N \varphi(\Pi_{n,0}^{-1}(U^i))
$$
then, to achieve an MSE of $\mathcal{O}(\epsilon^2)$ the cost would be $\mathcal{O}(n\epsilon^{-2})$ which is considerably larger if  $n$ is large. That is, the cost 
of the ML approach is essentially $\mathcal{O}(1)$ w.r.t.\ $n$.
If one stops at $n^*=|\log(\epsilon)/\log(\rho)|$ and uses the estimate
$$
\frac{1}{N}\sum_{i=1}^N \varphi(\Pi_{n^*,0}^{-1}(U^i))
$$
 to achieve an MSE of $\mathcal{O}(\epsilon^2)$, the cost is $\mathcal{O}(\epsilon^{-2}|\log(\epsilon)|)$.
A similar approach can show that these results are even true when smoothing for $\mathbb{E}[\varphi(X_{0:k})|y_{0:n}]$ for $k$ fixed (and hence $\mathbb{E}[\varphi(X_{s:s+k})|y_{0:n}]$).
The strategy of choosing $n^*$ and $N_{0:n^*}$ detailed above, is the one used throughout the paper. Note that in practice, we do not know $\rho$, so we choose a value such as $\rho=0.8$  which should lead to an $n^*$ which is large enough. This is also the reason for setting $N_p=\epsilon^{-2}(p+1)^{-1-\delta}$ and
not $N_p=\epsilon^{-2}(\rho^{1/2})^{p-1}$ say. 

It is remarked that the compactness of $\mathsf{X}$ could be removed by using Kellerer's  extension of the Kantorovich-Rubenstein theorem (see \cite{edwards} for a summary) and then, given that the latter theory is applicable, to show that
there exists a $C<+\infty$, $\rho\in(0,1)$ such that for any $n^*\geq p\geq 1$ 
$$
\sup_{\varphi \in \textrm{Lip}_1(\mathsf{X})'}|\mathbb{E}_{\pi_{p,0}}[\varphi(X_0)|y_{0:p}] - \mathbb{E}_{\pi_{p-1,0}}[\varphi(X_0)|y_{0:p-1}]| \leq C\rho^{p-1}
$$
where $\textrm{Lip}_1(\mathsf{X})'$ is the collection of functions $\varphi:\mathsf{X}\rightarrow\mathbb{R}$ such that for every $(x,y)\in\mathsf{X}^2$, 
$|\varphi(x)-\varphi(y)|\leq |x-y|^2$. This can be achieved using the techniques in \cite{jasra_smooth}.
Such an extension is mainly of a technical nature and is not required in the continuing exposition.
We now establish that the construction here can be extended to the case $\mathsf{X}\subset\mathbb{R}^d$.

\subsection{Case $\mathsf{X}\subset\mathbb{R}^d$}

We consider the Knothe-Rosenblatt rearrangement, which is assumed to exist (see e.g.~\cite{span}). For simplicity of notation, we set $\mathsf{X}=\mathsf{E}^d$ for some compact $\mathsf{E}\subset\mathbb{R}$.
Denote by $\Pi_{p,0}(\cdot|x_{1:j})$ the conditional CDF of $\pi_{p,0}(x_{j+1}|x_{1:j})$ with $1\leq j\leq d-1$. 
Note that here we are dealing with the $d-$dimensional co-ordinate at time zero and we considering conditioning on the first $j$ of these dimensions.
Then to approximate $\mathbb{E}_{\pi_{p,0}}[\varphi(X_0)|y_{0:p}] - \mathbb{E}_{\pi_{p-1,0}}[\varphi(X_0)|y_{0:p-1}]$, 
sample $U_{1:d}^1,\dots,U_{1:d}^{N_p}$, where for $i\in\{1,\dots,N_p\}$, $U_{1:d}^i\stackrel{\textrm{i.i.d.}}{\sim}\mathcal{U}_{[0,1]^d}$. Then we have the estimate for $\varphi\in\mathcal{B}_b(\mathsf{X})\cap\textrm{Lip}(\mathsf{X})$
$$
\frac{1}{N_p}\sum_{i=1}^{N_p} [\varphi(\xi_{p,d}^i)-\varphi(\xi_{p-1,d}^i)]
$$
where for ease of notation, we have set $\xi_{p,1}^i = \Pi_{p,0}^{-1}(U_1^i)$, (resp.\ $\xi_{p-1,1}^i = \Pi_{p-1,0}^{-1}(U_1^i)$) and $\xi_{p,j}^i = (\xi_{p,1}^i,\dots,\xi_{p,j-1}^i,\Pi_{p,0}^{-1}(U_j^i|\xi_{p,j-1}^i))$, $2\leq j \leq d$,  (resp.\ $\xi_{p-1,j}^i 
= (\xi_{p-1,1}^i,\dots,\xi_{p-1,j-1}^i,$ $\Pi_{p-1,0}^{-1}(U_j^i|\xi_{p-1,j-1}^i))$, $2\leq j \leq d$). We have the following result, whose proof and assumptions are in the appendix.

\begin{theorem}\label{theo:thm2}
Assume  (A\ref{ass:1}-\ref{ass:2}). Then there exists $\rho\in(0,1)$, $C<+\infty$ such that for any $\varphi\in\mathcal{B}_b(\mathsf{X})\cap\textrm{\emph{Lip}}(\mathsf{X})$, 
$n^*\geq p\geq 1$, $N_p\geq 1$, 
we have
$$
\mathbb{V}\textrm{\emph{ar}}\Big[
\frac{1}{N_p}\sum_{i=1}^{N_p} [\varphi(\xi_{p,d}^i)-\varphi(\xi_{p-1,d}^i)]
\Big] \leq \frac{C\rho^{p-1}\|\varphi\|_{\textrm{\emph{Lip}}}^2}{N_p}.
$$
\end{theorem}

We end this section with some remarks. Firstly, the MLMC strategy could be debiased w.r.t.\ the time parameter using the trick in \cite{rg:15}, which is a straightforward extension. One minor issue with this methodology, is that the variance
can blow up in some scenarios. Secondly, the idea of using the approach in \cite{rg:15}, when approximating  $\mathbb{E}[\varphi(X_{0;n})|y_{0:n}]$ has been adopted in \cite{jacob}. The authors use a conditional version of the coupled particle filter
(e.g.~\cite{chopin3,mlpf}) to couple smoothers, versus the optimal Wasserstein coupling. The goal in \cite{jacob} is unbiased estimation which is complementary to ideas in this article, where we focus upon reducing the cost of large lag smoothing.

\section{Transport methodology}\label{sec:method}

\subsection{Standard Approach}\label{sec:transportMethodology}

The basic principle of the transport methodology introduced in \cite{span} is to determine a mapping $T$ relating a base distribution $\eta$, e.g.\ the normal distribution, to a potentially sophisiticated target distribution $\tilde\pi$ related to the problem of interest. The distribution $\eta$ should be easy to sample from so that, given the map $T$, we can obtain samples from $\tilde\pi$ by simply mapping samples from $\eta$ via $T$. More precisely, the considered mapping $T$ is characterised by
$$
T_{\#} \eta(x) = \eta(T^{-1}(x)) |\det \nabla T^{-1}(x)| = \tilde\pi(x),
$$
that is, the \emph{push-forward} distribution of $\eta$ by $T$ is $\tilde\pi$. Such a mapping can be approximated using deterministic or stochastic optimisation methods. However, the underlying optimisation problem is only amenable when the space on which $\tilde\pi$ is defined is of a low dimension, e.g.\ up to $4$. This is not the case in general for the smoothing distributions introduced in the previous sections, especially as the number of observations increases. This is addressed in \cite{span} by identifying the dependence structure between the random variables of interest. In particular, for a hidden Markov model on $\bbR^d$, it is possible to decompose the problem into transport maps of dimension~$2d$, which does not depend on the number $n$ of observations that define the smoother. The problem at time $p$ can be solved by introducing a mapping $T_p$ of the form
$$
T_p(x_p,x_{p+1}) =
\begin{bmatrix}
T^0_p(x_p, x_{p+1}) \\
T^1_p(x_{p+1})
\end{bmatrix}
$$
which will transform the $2d$-dimensional base distribution $\eta_{2d}$ into a target distribution related to the considered hidden Markov model, as detailed below. This target distribution can be expressed as
$$
\tilde\pi_p(x_p,x_{p+1}) \propto \eta_d(x_p) f\big(T^1_{p-1}(x_p), x_{p+1}\big) g(x_{p+1}, y_{p+1}),
$$
for any $p > 0$, which can be seen to be the 1-lag smoother. When $p=0$, we simply define $\tilde\pi_0(x_0,x_1) = f(x_0) f(x_0, x_1) g(x_0, y_0) g(x_1, y_1)$. The base distribution $\eta_{2d}$ (resp.\ $\eta_d$) is the standard normal distribution of dimension $2d$ (resp.\ $d$). The mapping $T_p$ can be embedded into the $2d(n+1)$-dimensional identity mapping as
$$
\bar{T}_p(x_0,\dots,x_n) = (x_0, \dots, x_{p-1}, T^0_p(x_p, x_{p+1}), T^1_p(x_{p+1}), x_{p+2}, \dots, x_n)^{\tr},
$$
with $\cdot^{\tr}$ denoting the matrix transposition. It follows that
$$
\bm{T}_n = \bar{T}_0 \circ \dots \circ \bar{T}_n
$$
is the map such that the pushforward $(\bm{T}_n)_{\#} \eta_{d(n+1)}$ is equal to the probability density function of the smoother at time $n$. Obtaining samples from the smoothing distribution is then straightforward: it suffices to sample from $\eta_{d(n+1)}$ and to map the obtained sample via~$\bm{T}_n$.

Even in low dimension, the optimisation problem underlying the computation of the transport maps of interest is not trivial. One first has to consider an appropriate parametrisation of these maps, e.g.\ via polynomial representations. The parameters of the considered representation then have to be determined using the following optimisation problem
\begin{equation}\label{eq:optimisation}
T_{p}^{*} = \argmin_T - \mathbb{E}\bigg[ \log \tilde\pi_p (T(X)) + \log \big(\det \nabla T(X)\big) - \log \eta_{2d}(X)  \bigg],
\end{equation}
where the minimum is taken over the set of monotone increasing lower-triangular maps. This minimisation problem can be solved numerically by considering a parametrised family of maps and deterministic or stochastic optimisation methods. Let $T$ be any acceptable map in the minimisation \eqref{eq:optimisation} and denote by $T^{(i)}$ the $i$\textsuperscript{th} component of $T$, which only depends on the $i$\textsuperscript{th} first variables, $i \in \{1,\dots,2d\}$, then the considered parametrisation can be expressed as
$$
T^{(i)}(x_1,\dots,x_i) = a_i(x_1,\dots,x_{i-1}) + \int_0^{x_i} b_i(x_1, \dots, x_{i-1},t)^2 d t
$$
for some real-valued functions $a_i$ and $b_i$ on $\bbR^{i-1}$ and $\bbR^i$ respectively. It is assumed that the functions $x_j \mapsto a_i(x_1,\dots,x_{i-1})$ and $x_j \mapsto b_i(x_1,\dots,x_{i-1},t)$ are Hermite Probabilists' functions extended with constant and linear components for any $j \leq i-1$, and the function $t \mapsto b_i(x_1,\dots,x_{i-1},t)$ is also a Hermite Probabilists' function which is only extended with a constant component. In particular, these functions take the form
\begin{align*}
a_i(x_1,\dots,x_{i-1}) & = \sum_{k = 1}^{2d(o_{\mathrm{map}}+1)} c_k \Phi_k(x_1,\dots,x_{i-1}) \\
b_i(x_1,\dots,x_{i-1},t) & = \sum_{k = 1}^{2do_{\mathrm{map}}} c'_k \Psi_k(x_1,\dots,x_{i-1},t)
\end{align*}
with $o_{\mathrm{map}}$ the map order, with $\{c_k\}_{k \geq 1}$ and $\{c'_k\}_{k \geq 1}$ some collections of real coefficients and with $\Phi_k$ and $\Psi_k$ basis functions based on the above mentioned Hermite Probabilists' functions. The expectation in \eqref{eq:optimisation} is then approximated using a Gauss quadrature of order $o_{\mathrm{exp}}$ in each dimension and the minimisation is solved via the Newton algorithm using the conjugate-gradient method for each step.

The desired function $T_p$ can be recovered through the relation
$$
T_p((x_{p,1},\dots,x_{p,d}),(x_{p+1,1}, \dots, x_{p+1,d})) = (S_{\sigma} \circ T^*_p \circ S_{\sigma})(x_{p,1},\dots,x_{p,d}, x_{p+1,1}, \dots, x_{p+1,d}),
$$
where $\sigma = (2d, 2d-1, \dots, 1)$ and $S_{\sigma}$ is the linear map corresponding to the permutation matrix of $\sigma$, which verifies $S_{\sigma}^{-1} = S_{\sigma}$.

\subsection{Fixed-Point Smoothing with Transport Maps}

The approach described in Section \ref{sec:transportMethodology} allows for obtaining samples from the distribution $\pi_{n,0}$ of $X_0$ given $(Y_0,\dots,Y_n) = (y_0,\dots,y_n)$ by simply retaining the first $d$ components of samples from $\eta_{d(n+1)}$ after mapping them through $\bm{T}_n$. However, the computational cost associated with the mapping of samples by $\bm{T}_n$ increases with $n$, making the complexity of the method of the order $\mathcal{O}(n^2)$.

This can however be addressed by considering $X_0$ as a parameter and by only propagating the transport map corresponding to the posterior distribution of $(X_0,X_n)$. This approach has been suggested in \cite[section~7.4]{span}. We assume in the remainder of this section that observations start at time step $1$ instead of $0$. When considering $X_0$ as a parameter, the elementary transport maps take the form
$$
T_p(x_0,x_p,x_{p+1}) =
\begin{bmatrix}
T^{X_0}_p(x_0) \\
T^0_p(x_0, x_p, x_{p+1}) \\
T^1_p(x_0, x_{p+1})
\end{bmatrix}.
$$
and the corresponding target distributions become
$$
\tilde\pi_1(x_0, x_1, x_2) \propto p_0(x_0) f(x_0, x_1) f(x_1, x_2) g(x_1, y_1) g(x_2, y_2),
$$
and
$$
\tilde\pi_p(x_0, x_p, x_{p+1}) \propto \eta_{2d}(x_0, x_p) f\big(T^1_{p-1}(x_0, x_p), x_{p+1} \big) g(x_{p+1}, y_{p+1}),
$$
for any $p > 1$. The transport map associated with the posterior distribution of $(X_0,X_n)$ is
$$
\hat{T}_n(x_0, x_n) = 
\begin{bmatrix}
T^{X_0}_1 \circ \dots \circ T^{X_0}_{n-1}(x_0) \\
T^1_{n-1}(x_0, x_n)
\end{bmatrix}.
$$
By recursively approximating the composition $T^{X_0}_1 \circ \dots \circ T^{X_0}_{n-1}$ by a single map, the computation of samples from the posterior distribution of $X_0$ becomes linear in time. The pseudo-code for this approach is given in Algorithm~\ref{alg:multilevel_transport}.

\begin{algorithm}
\caption{Multilevel transport}
\label{alg:multilevel_transport}
\begin{algorithmic}[1]
\State \emph{input:} $\epsilon$, $\delta$, $\rho$ 
\State \emph{Output:} estimate $\hat{X}_0$ of $X_0 \mid y_{0:n^*}$
\State $n^* = \log(\epsilon)/\log(\rho)$
\For {$p=1,\dots,n^*$}
\If{$p = 1$}
\State $\tilde\pi_p(x_0, x_1, x_2) \propto p_0(x_0) f(x_0, x_1) f(x_1, x_2) g(x_1, y_1) g(x_2, y_2)$
\Else 
\State $\tilde\pi_p(x_0, x_p, x_{p+1}) \propto \eta_{2d}(x_0, x_p) f\big(T^1_{p-1}(x_0, x_p), x_{p+1} \big) g(x_{p+1}, y_{p+1})$
\State
\Comment{$T^1_{p-1}$ is the second component of $\hat{T}_{p-1}$}
\EndIf
\State $\eta = \mathcal{N}(\mathbf{0}_{2d}, \mathbf{I}_{2d})$
\State $\hat{T}_p = \text{FilteringDistributionTransportMap}(\eta, \tilde\pi_p)$
\State 
\Comment{Compute transport map from $\eta$ to the law of $(X_0,X_p) \mid y_{1:p}$ based on $\tilde\pi_p$}
\State $N_p = \epsilon^2 (p+1)^{-1-\delta}$
\Comment{Compute the number of samples}
\For {$i=1,\dots,N_p$}
\State $S \sim \eta$
\State $\xi^i_p = \hat{T}_p(S)$
\If{$p = 1$}
\State $\zeta^i_p = \varphi(\xi^{i,1:d}_p)$
\Comment{Map the first $d$ components of $\xi^i_p$ through $\varphi$}
\Else
\State $\xi^i_{p-1} = \hat{T}_{p-1}(S)$
\State $\zeta^i_p = \varphi(\xi^{i,1:d}_p) - \varphi(\xi^{i,1:d}_{p-1})$
\EndIf
\EndFor
\State $\hat{X}_0 \leftarrow \hat{X}_0 + \frac{1}{N_p}\sum_{i=1}^{N_p} \zeta^i_p$
\EndFor
\end{algorithmic}
\end{algorithm}

\section{Case Studies}\label{sec:case}

\subsection{Linear Gaussian}

\subsubsection{Theoretical Result}

The results in Section \ref{sec:approach} do not apply to the linear Gaussian case. We extend our results to this scenario.
We assume that the dynamical and observations models are one-dimensional as well as linear and Gaussian such that the state and observation random variables at time $n$ can be defined as
\begin{eqnarray*}
X_n|x_{n-1} & \sim & \calN( \alpha x_{n-1}, \beta^2), \qquad  n\geq 1\\
Y_n|x_n & \sim & \calN( x_n, \tau^2 ), \qquad  n\geq 0
\end{eqnarray*}
and $X_0 \sim \mathcal{N}(0, \sigma^2)$, for some $\alpha \in \bbR$ and some $\beta, \sigma, \tau > 0$. We have the following result, whose proof is in the appendix.

\begin{theorem}
\label{thm:linearGaussian}
Assuming that $\mathbb{V}\textrm{\emph{ar}}( X_p \mid y_{0:p} ) \approx \gamma^2$ for all $p$ large enough, it holds that
$$
\mathbb{V}\textrm{\emph{ar}}\bigg[ \dfrac{1}{N_p} \sum_{i=1}^{N_p} [ \Pi_{p,0}^{-1}(U^i) - \Pi_{p-1,0}^{-1}(U^i) ] \bigg] = \mathcal{O}\bigg( \dfrac{1}{N_p} \Big(\alpha + \frac{\beta^2}{\alpha \gamma^2} \Big)^{-2p} \bigg).
$$
\end{theorem}

Theorem \ref{thm:linearGaussian} shows that, under assumptions on the parameters of the model, the variance of the approximated multilevel term at level $p$ tends to $0$ exponentially fast in $p$ and with an order of $1/N_p$ for the number of samples.  This theorem also indicates that the behaviour depends an all the parameters in the model, although implicitly in $\tau$. For instance, if $\beta \gg \tau$ then one can consider $\gamma = \tau$ in the above expression. The assumption about the variance of the filter can be justified in terms of reachability and observability of the system \cite{Kumar}.

This rate can get extremely beneficial for the proposed approach when $\beta$ is large and $\gamma$ is small, however it can also make it of little use in the opposite case. This does not come as a surprise since a large $\beta$ means that the initial condition is quickly forgotten so that obtaining a high number of samples from the smoother $\pi_{p,0}$ for large $p$ would be inefficient, whereas small values of $\beta$ incur a much higher dependency between the initial state and the observations at different time steps.

\subsubsection{Numerical Results}

The performance of the proposed method is first assessed in the linear-Gaussian case where an analytical solution of the fixed-point smoothing problem is available, this solution being known as the Rauch-Tung-Striebel smoother \cite{rauch}. More specifically, we consider the following model:
\begin{eqnarray*}
X_n | x_{n-1} & \sim & \mathcal{N}(\alpha x_{n-1}, \beta^2), \qquad  n\geq 1 \\
Y_n | x_n & \sim  & \mathcal{N}(x_n, \tau^2), \qquad  n\geq 0
\end{eqnarray*}
with $X_0 \sim \mathcal{N}(1, \sigma^2)$, where $\sigma = 2$ and $\alpha = \beta = \tau = 1$. The transport maps of interest are approximated to the order $o_{\mathrm{map}} = 3$ while the expectation is approximated to the order $o_{\mathrm{exp}} = 5$ and the minimisation is performed with a tolerance of $10^{-4}$. The number of samples at each time step as well as the time horizon $n^*$ is computed according to the method proposed in 
Section \ref{sec:case_R} with different values for the parameter $\epsilon$. The performance of the proposed method is compared against the PaRIS algorithm introduced in \cite{olsson} using the observations $y_1,\dots,y_{25}$ with a varying number $N$ of samples and with $\tilde{N} = 2$ terms for the propagation of the estimate of $X_0$. In the simulations, it always holds that $n^* \leq 25$ to ensure the fairness of the comparison. The criteria for performance assessment is the MSE at the final time step, defined as
$$
\dfrac{1}{M} \sum_{i=1}^M (\hat{x}_i - x^*)^2
$$
where $M$ is the number of Monte Carlo simulations, $\hat{x}_i$ is the estimate of $X_0 \mid y_{1:n^*}$ (with $n^* = 25$ for the PaRIS algorithm) and where $x^*$ is the corresponding estimate given by the Rauch-Tung-Striebel smoother.

The values of the MSE at the final time obtained in simulations are shown in Figure \ref{fig:linearGaussian} where the proposed approach displays smaller errors than the PaRIS algorithm for different values of $\epsilon$ and $N$. The advantage when representing the probability distributions of interest with transport maps is that the computational effort required to obtain a sample is extremely limited once the maps have been determined. For instance, the highest and lowest considered values of $\epsilon$ in Figure \ref{fig:linearGaussian} correspond to $N_1 = 1250$ and $N_1 = 500,000$ samples respectively, which induces a comparatively small increase in computational time.

\begin{figure}
\centering
\includegraphics[width=.75\textwidth]{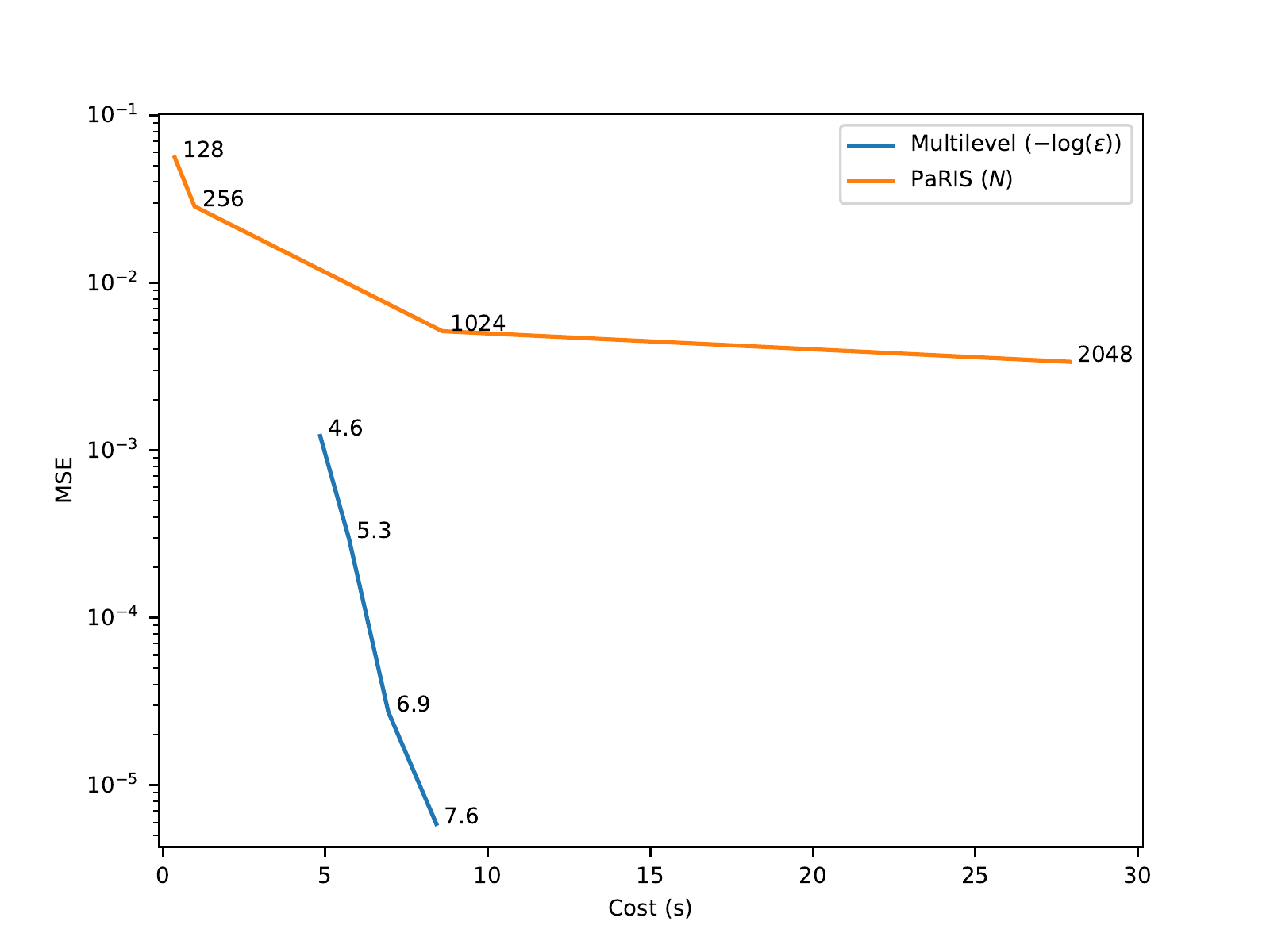}
\caption{Performance of the proposed method against the PaRIS algorithm with the linear-Gaussian model, averaged over 100 Monte Carlo simulations. The reference for the computation of the MSE is the Rauch-Tung-Striebel smoother. The displayed cost for the multilevel approach includes the computation of the transport maps.}
\label{fig:linearGaussian}
\end{figure}

\subsection{Stochastic Volatility Model}

In order to further demonstrate the performance of the proposed approach, the assessment conducted in the previous section is applied to a non-linear case. A stochastic volatility model is considered with
\begin{eqnarray*}
X_n & = & \mu + \phi(X_{n-1} - \mu) + V_n, \qquad  n\geq 1, \qquad X_0 \sim \mathcal{N}\Big(\mu, \dfrac{1}{1-\phi^2}\Big) \\
Y_n & =&  W_n \exp\Big(\frac{1}{2}X_n \Big), \qquad n\geq 0
\end{eqnarray*}
with $V_n \sim \calN(0,\beta^2)$ and $W_n \sim \calN(0,1)$, where $\mu = -0.5$, $\phi = 0.95$ and $\beta = 0.25$. In the absence of an analytical solution, the reference is determined by the PaRIS algorithm with $N = 2^{13}$ samples. Since the observation process of this model is generally less informative than the one of the Gaussian model, the PaRIS algorithm is given the observations up to the time step $50$ and, similarly, it is ensured that $n^* \leq 50$ for the proposed approach. The other parameters are the same as in the linear-Gaussian case.

The MSE at the final time obtained for the two considered methods is shown in Figure~\ref{fig:stochasticVolatility}. Once again, the error for the proposed approach is lower than for the PaRIS algorithm although the difference is less significant. In particular, the gain in accuracy between the lowest and the second lowest value of $\epsilon$ seem to indicate that simply increasing the number of samples would not allow for reducing the error much further. However, increasing the order of the transport maps or decreasing the tolerance in the optimisation could further reduce the error, although with a significantly higher computational cost.

\begin{figure}
\centering
\includegraphics[width=.75\textwidth]{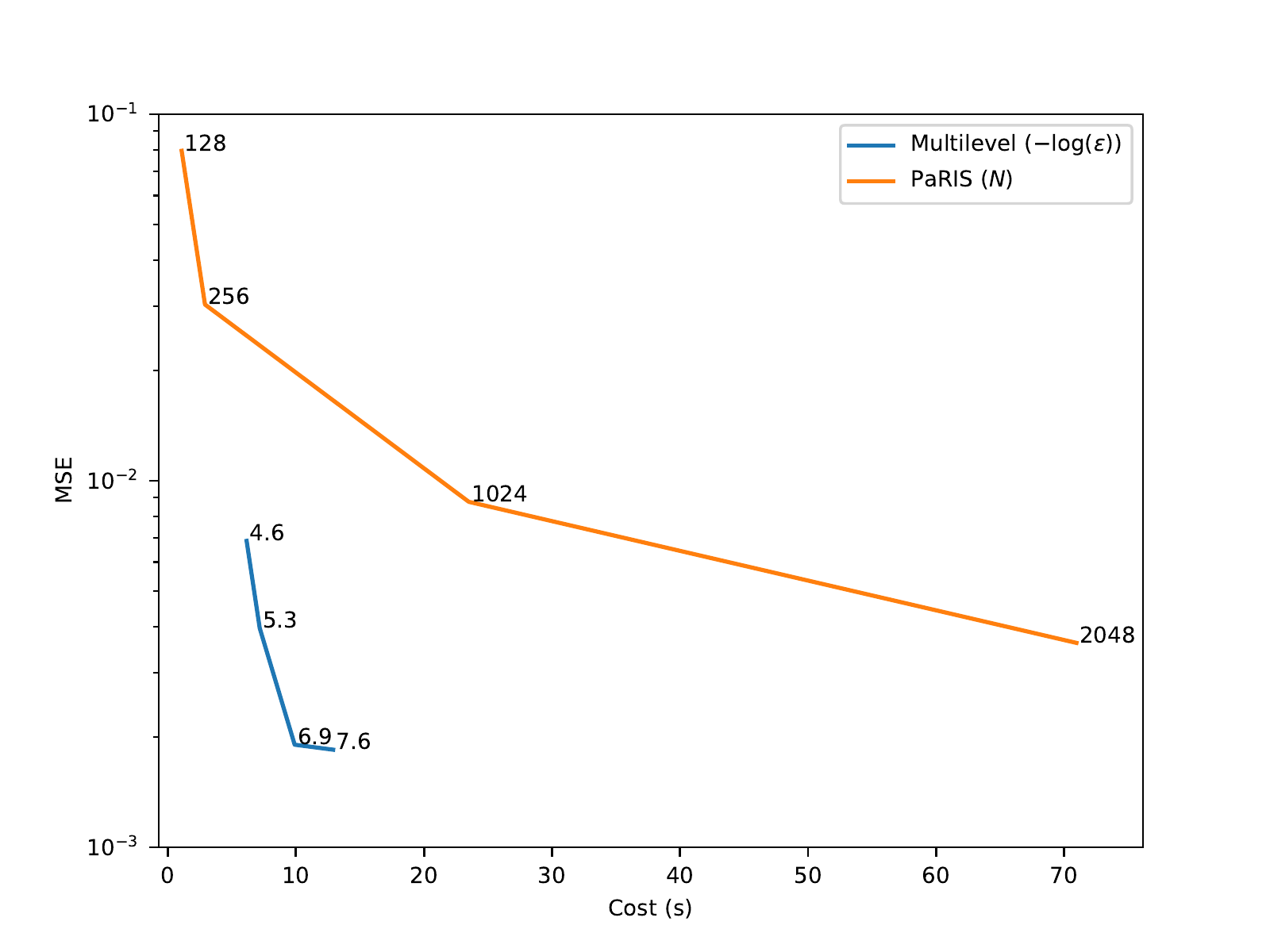}
\caption{Performance of the proposed method against the PaRIS algorithm with the stochastic volatility model, averaged over 100 Monte Carlo simulations. The reference for the computation of the MSE is the PaRIS algorithm with $2^{13}$ samples. The displayed cost for the multilevel approach includes the computation of the transport maps.}
\label{fig:stochasticVolatility}
\end{figure}

The computational costs obtained for the two models considered in simulations are shown in Figure \ref{fig:complexity} for different values of $\epsilon$. These results confirm the order $\mathcal{O}(\epsilon^{-2})$ that was predicted in Section \ref{sec:approach}.

\begin{figure}
\centering
\includegraphics[width=.75\textwidth]{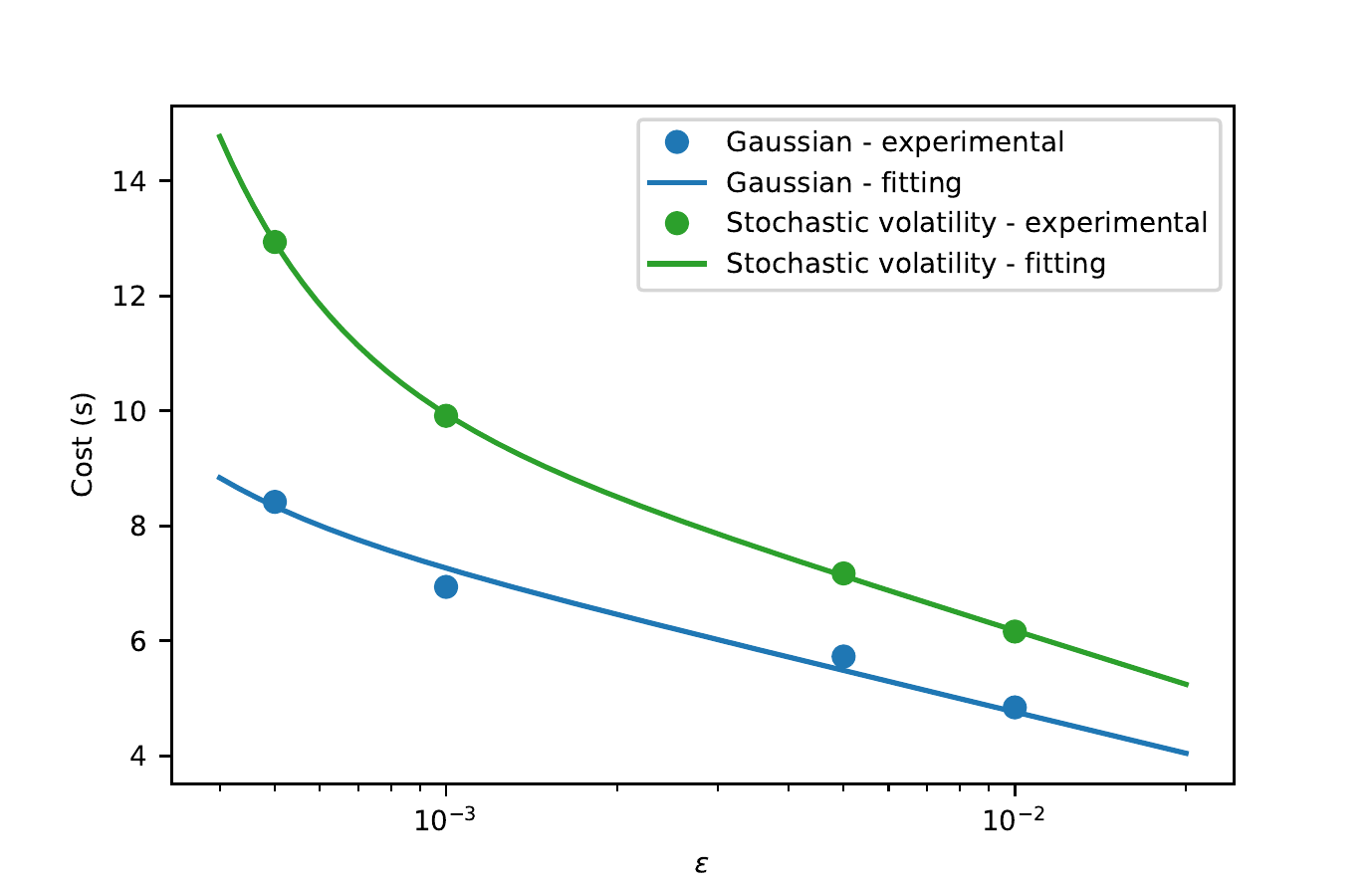}
\caption{Computational cost as a function of $\epsilon$, averaged over 100 Monte Carlo simulations. The fitted curves are based on a function of the form $\epsilon \mapsto -a \epsilon^{-2} - b \log(\epsilon)$, with $a$ and $b$ some parameters, which is justified by the form of the cost \eqref{eq:cost_ml}.}
\label{fig:complexity}
\end{figure}

\section{Summary}\label{sec:summ}

In this article we have considered large lag smoothing for HMMs, using the MLMC method. We
showed that under an optimal coupling when the hidden state is in dimension 1 or higher, but
on a compact space that, essentially, the cost can be decoupled from the time parameter of the smoother. As this optimal method is not possible in practice, we showed how it could be approximated and established numerically that our theory still holds in this approximated case. Several extensions to the work are possible. Firstly, to extend our theoretical results to the case of the approximated coupling. Secondly, to investigate whether the coupling used in \cite{jacob} can also yield, theoretically, the same improvements that have been seen in the work in this article.

\subsection*{Acknowledgements}

All authors were supported by Singapore Ministry of Education AcRF tier 1 grant R-155-000-182-114. AJ is affiliated with the Risk Management Institute, OR and analytics cluster and the Center for Quantitative Finance at NUS.

\appendix

\section{Variance Proofs}

We write the density (or probability measure) of the smoother, at time $p$, on the co-ordinate at time zero as $\pi_{p,0}$ and the associated CDF as $\Pi_{p,0}$ (with generalized inverse $\Pi_{p,0}^{-1}$).
Recall that throughout $\mathsf{X}$ is a compact subspace of $\mathbb{R}^d$. 
Throughout the observations are fixed and often omitted from the notations.
The appendix gives our main assumptions, followed by a technical Lemma (Lemma \ref{lem:lem1}) which
features some technical results used in the proofs. Then the proof of Theorem \ref{theo:thm1} is given. The appendix is concluded by a second technical Lemma (Lemma \ref{lem:lem2})
followed by the proof of Theorem \ref{theo:thm2}.

\begin{hypA}\label{ass:1}
There exists $0<\underline{C}<\overline{C}<+\infty$ such that
\begin{eqnarray*}
\inf_{x\in\mathsf{X}}g(x,y_0)f(x)\wedge\inf_{p\geq 1}\inf_{(x,x')\in\mathsf{X}^2} g(x',y_p) f(x,x') & \geq &\underline{C} \\
\sup_{x\in\mathsf{X}}g(x,y_0)f(x)\vee\sup_{p\geq 1}\sup_{(x,x')\in\mathsf{X}^2} g(x',y_p) f(x,x') & \leq &\overline{C}.
\end{eqnarray*}
\end{hypA}
\begin{hypA}\label{ass:2}
There exists $C<+\infty$ such that for every $(x,x')\in\mathsf{X}^2$
\begin{eqnarray*}
|g(x,y_0)-g(x',y_0)| & \leq & C |x-x'| \\
\sup_{z\in\mathsf{X}}|f(x,z)-f(x',z)| & \leq & C |x-x'| \\
|f(x)-f(x')| & \leq & C |x-x'| .
\end{eqnarray*}
\end{hypA}

Below $\pi_{p,0}(\cdot|x_{1:j})$ denotes the probability of the $(j+1)^{th}$ co-ordinate of the smoother at time $0$, given the first $j-$co-ordinates at time 0, and conditional upon
the observations up-to time $p$.

\begin{lem}\label{lem:lem1}
Assume (A\ref{ass:1}-\ref{ass:2}).  Then there exists $(C,C')\in(0,\infty)^2$, $\rho\in (0,1)$ such that  
\begin{enumerate}
\item{for any $1\leq j \leq d$, $\sup_{p\geq 0}\pi_{p,0}(x_{0,1:j})\leq C$, $\inf_{p\geq 0}\pi_{p,0}(x_{0,1:j})\geq C'$}
\item{for any $p\geq 1$, $\|\pi_{p,0}-\pi_{p-1,0}\|_{\textrm{\emph{tv}}}\leq C \rho^{p-1}$}
\item{for any $1\leq j \leq d$, $p\geq 1$, $\sup_{x_{1:j}\in\mathsf{E}^j}\|\pi_{p,0}(\cdot|x_{1:j})-\pi_{p-1,0}(\cdot|x_{1:j})\|_{\textrm{\emph{tv}}}\leq C \rho^{p-1}$}
\item{for any $p\geq 0$, $(x,x')\in\mathsf{X}^2$, $|\pi_{p,0}(x)-\pi_{p,0}(x')| \leq C |x-x'|$}
\item{for any $p\geq 0$, $1\leq j \leq d$, $(x_{1:j},x_{1:j}')\in(\mathsf{E}^j)^2$, $|\pi_{p,0}(x_{1:j})-\pi_{p,0}(x_{1:j}')| \leq C |x_{1:j}-x_{1:j}'|$.}
\end{enumerate}
\end{lem}

\begin{proof}
1.\ follows trivially from (A\ref{ass:1}) and the compactness of $\mathsf{E}$. 2.\ follows from the backward Markov chain representation of the smoother
 and (A\ref{ass:1}); see for instance \cite{Cappe_2005} and the references therein.

3.\ to prove this result, we first consider controlling for any fixed $1\leq j \leq d$  $p\geq 1$,
$$
|\pi_{p,0}(x_{1:j})-\pi_{p-1,0}(x_{1:j})|.
$$
Denoting $\pi_{(p)}$ as the filter at time $p$ 
and setting for $k\geq 0$
$$
B_{k}(x_{k+1},x_k) = \frac{\pi_{(k)}(x_k)f(x_k,x_{k+1})}{\int_{\mathsf{X}}\pi_{(k)}(x_k)f(x_k,x_{k+1})dx_k}
$$
we can write
$$
|\pi_{p,0}(x_{1:j})-\pi_{p-1,0}(x_{1:j})| =
\textrm{Osc}(B_0(\cdot,x_{1:j}))\Big|[\pi_{(p)}B_{p-1}-\pi_{(p-1)}](B_{p-2:1})\Big(\frac{B_0(\cdot,x_{1:j})}{\textrm{Osc}(B_0(\cdot,x_{1:j}))}\Big)\Big|.
$$
Using standard results for the total variation distance
$$
|\pi_{p,0}(x_{1:j})-\pi_{p-1,0}(x_{1:j})| \leq \textrm{Osc}(B_0(\cdot,x_{1:j})) \prod_{s=1}^{p-2}\omega(B_s)
$$
where $\omega(B_s)$ is the Dobrushin coefficient of the Markov kernel $B_s$. Standard calculations yield that there exists a 
$\rho \in (0,1)$ such that $\textrm{Osc}(B_0(\cdot,x_{1:j}))\vee \omega(B_s) \leq C \rho$, where $C$ does not depend upon $x_{1:j}$.
Hence we have shown that 
\begin{equation}\label{eq:pref_eq1}
\sup_{x_{1:j}\in \mathsf{E}^j} |\pi_{p,0}(x_{1:j})-\pi_{p-1,0}(x_{1:j})| \leq C \rho^{p-1}.
\end{equation}
To prove the result of interest we have for any $\varphi\in\textrm{Osc}_1(\mathsf{E})$
\begin{eqnarray*}
|\pi_{p,0}(\varphi|x_{1:j})-\pi_{p-1,0}(\varphi|x_{1:j})|  & = &
\frac{1}{\pi_{p,0}(x_{1:j-1})}\int_{\mathsf{E}}\varphi(x_j)[\pi_{p,0}(x_{1:j})-\pi_{p-1,0}(x_{1:j})]dx_j + \\ & & 
\frac{\pi_{p-1,0}(x_{1:j-1})-\pi_{p,0}(x_{1:j-1})}{\pi_{p,0}(x_{1:j-1})\pi_{p-1,0}(x_{1:j-1})}
\int_{\mathsf{E}}\varphi(x_j)\pi_{p-1,0}(x_{1:j})dx_j.
\end{eqnarray*}
The conclusion then follows by using \eqref{eq:pref_eq1} and 1..

4.\ follows almost immediately from (A\ref{ass:2}) and the definition of the smoother. 5.\ follows from 4.\ on marginalization
and the compactness of $\mathsf{E}$.
\end{proof}

\begin{proof}[Proof of Theorem \ref{theo:thm1}]
Standard calculations for i.i.d.\ random variables and the Lipschitz property of $\varphi$ clearly yields:
$$
\mathbb{V}\textrm{ar}\Big[\frac{1}{N_p}\sum_{i=1}^N [\varphi(\Pi_{p,0}^{-1}(U^i))-\varphi(\Pi_{p-1,0}^{-1}(U^i))]\Big] \leq
 \frac{\|\varphi\|_{\textrm{Lip}}^2}{N_p} \int_{[0,1]}|\Pi_{p,0}^{-1}(u)-\Pi_{p-1,0}^{-1}(u)|^2du.
$$

Now we note that
$$
\int_{[0,1]}|\Pi_{p,0}^{-1}(u)-\Pi_{p-1,0}^{-1}(u)|^2du = W_2(\pi_{p,0},\pi_{p-1,0})^2
$$
where $W_2(\pi_{p,0},\pi_{p-1,0})$ is the 2-Wasserstein distance between $\pi_{p,0}$ and $\pi_{p-1,0}$.
As $\mathsf{X}$ is compact it follows
$$
W_2(\pi_{p,0},\pi_{p-1,0})^2 \leq \Big(\int_{\mathsf{X}}dx\Big)^2 \|\pi_{p,0}-\pi_{p-1,0}\|_{\textrm{tv}}
$$
where $\|\cdot\|_{\textrm{tv}}$ is the total variation distance. Under our assumptions one can show that 
there exists $\rho\in(0,1)$, $C<+\infty$ such that for any $p\geq 1$ (see Lemma \ref{lem:lem1} 2., which holds when $d=1$)
$$
\|\pi_{p,0}-\pi_{p-1,0}\|_{\textrm{tv}} \leq C \rho^{p-1}.
$$
The proof is then easily concluded.
\end{proof}

\begin{lem}\label{lem:lem2}
Assume (A\ref{ass:1}-\ref{ass:2}).  Then there exists $C<+\infty$, $\rho\in (0,1)$ such that for any $p\geq 1$ 
$$
\mathbb{E}[|\xi_{p,d}^1-\xi_{p-1,d}^1|^2] \leq C \rho^{p-1}.
$$
\end{lem}

\begin{proof}
The proof is by induction on $d$, the case $d=1$ being proved by the approach in the proof of Theorem \ref{theo:thm1}.
Throughout $C$ is a finite constant whose value may change from line-to-line, but does not depend upon $p$.

We suppose the result for $d-1$ and consider $d$.  
For simplicity of notation, we drop the superscript 1 from the notation, e.g.\ we write $\xi_{p,d}$ instead of $\xi_{p,d}^1$.
We have
\begin{eqnarray}
\mathbb{E}[|\xi_{p,d}-\xi_{p-1,d}|^2] & = & \mathbb{E}[\mathbb{E}[|\xi_{p,d}^1-\xi_{p-1,d}^1|^2|U_{1:d-1}]] \nonumber\\
& \leq & C \mathbb{E}[\|\pi_{p,0}(\cdot|\xi_{p,d-1}) - \pi_{p-1,0}(\cdot|\xi_{p-1,d-1})\|_{\textrm{tv}}] \label{eq:master_lem_rd}
\end{eqnarray}
where, to go to the second line, we have used (conditional upon $U_{1:d}$) the relationship between the squared 2-Wasserstein distance and the (generalized) inverse CDF,
along with the total variation bound as used in the proof of Theorem \ref{theo:thm1}.

Now, we have
$$
\|\pi_{p,0}(\cdot|\xi_{p,d-1}) - \pi_{p-1,0}(\cdot|\xi_{p-1,d-1})\|_{\textrm{tv}} \leq 
$$
\begin{equation}\label{eq:prf2_eq1}
\|\pi_{p,0}(\cdot|\xi_{p,d-1}) - \pi_{p-1,0}(\cdot|\xi_{p,d-1})\|_{\textrm{tv}} + 
\|\pi_{p-1,0}(\cdot|\xi_{p,d-1}) - \pi_{p-1,0}(\cdot|\xi_{p-1,d-1})\|_{\textrm{tv}}.
\end{equation}
By Lemma \ref{lem:lem1} 3.\ it follows that 
\begin{equation}\label{eq:prf2_eq2}
\|\pi_{p,0}(\cdot|\xi_{p,d-1}) - \pi_{p-1,0}(\cdot|\xi_{p,d-1})\|_{\textrm{tv}} \leq C \rho^{p-1}
\end{equation}
so we consider $\|\pi_{p-1,0}(\cdot|\xi_{p,d-1}) - \pi_{p-1,0}(\cdot|\xi_{p-1,d-1})\|_{\textrm{tv}}$. For any $\varphi\in\textrm{Osc}_1(\mathsf{E})$
\begin{multline*}
\pi_{p,0}(\varphi|\xi_{p,d-1}) - \pi_{p-1,0}(\varphi|\xi_{p,d-1}) = \\
\frac{1}{\pi_{p-1,0}(\xi_{p-1,d-1})}\int_{\mathsf{E}}\varphi(x)[\pi_{p-1,0}(\xi_{p,d-1},x)-\pi_{p-1,0}(\xi_{p-1,d-1},x)]dx + \\
\frac{\pi_{p-1,0}(\xi_{p-1,d-1})-\pi_{p-1,0}(\xi_{p,d-1})}{\pi_{p-1,0}(\xi_{p,d-1})\pi_{p-1,0}(\xi_{p-1,d-1})}
\int_{\mathsf{E}}\varphi(x) \pi_{p-1,0}(\xi_{p-1,d-1},x)dx.
\end{multline*}
Applying Lemma~\ref{lem:lem1} 4.\ to the first term on the R.H.S.\ and Lemma~\ref{lem:lem1} 5.\ to the second term on the R.H.S.\ along with the boundedness of $\varphi$ and compactness of $\mathsf{E}$, we have that
\begin{eqnarray*}
|\pi_{p,0}(\varphi|\xi_{p,d-1}) - \pi_{p-1,0}(\varphi|\xi_{p,d-1})|  & \leq & \frac{C}{\pi_{p-1,0}(\xi_{p-1,d-1})}|\xi_{p,d-1}-\xi_{p-1,d-1}|
+  \\ & &
\frac{C}{\pi_{p-1,0}(\xi_{p,d-1})\pi_{p-1,0}(\xi_{p-1,d-1})}|\xi_{p,d-1}-\xi_{p-1,d-1}|.
\end{eqnarray*}
Applying Lemma~\ref{lem:lem1} 1.\ we can then establish that
\begin{equation}\label{eq:prf2_eq3}
\|\pi_{p-1,0}(\cdot|\xi_{p,d-1}) - \pi_{p-1,0}(\cdot|\xi_{p-1,d-1})\|_{\textrm{tv}} \leq C |\xi_{p,d-1}-\xi_{p-1,d-1}|.
\end{equation}
Combining \eqref{eq:prf2_eq2} and \eqref{eq:prf2_eq3} with \eqref{eq:prf2_eq1} and noting \eqref{eq:master_lem_rd}, we have shown that
$$
\mathbb{E}[|\xi_{p,d}-\xi_{p-1,d}|^2] \leq C\Big(\rho^{p-1} + \mathbb{E}[|\xi_{p,d-1}-\xi_{p-1,d-1}|]\Big).
$$
The proof is completed by using the Jensen inequality and the induction hypothesis.
\end{proof}

\begin{proof}[Proof of Theorem \ref{theo:thm2}]
We have
$$
\mathbb{V}\textrm{ar}\Big[
\frac{1}{N_p}\sum_{i=1}^{N_p} [\varphi(\xi_{p,d}^i)-\varphi(\xi_{p-1,d}^i)]
\Big] \leq \frac{\|\varphi\|_{\textrm{Lip}}^2}{N_p}\mathbb{E}[|\xi_{p,d}^1-\xi_{p-1,d}^1|^2] .
$$
The proof is then completed by applying Lemma \ref{lem:lem2}.
\end{proof}

\section{Linear Gaussian Result}

\begin{proof}[Proof of Theorem \ref{thm:linearGaussian}]
The Rauch-Tung-Striebel smoother gives an expression of the smoothed mean $m_{p|n}$ and variance $v_{p|n}$ at time $p$ given the observations $y_0,\dots,y_n$ as
\begin{eqnarray*}
m_{p|n} & = & m_{p|p} + c_p( m_{p+1|n} - m_{p+1|p} ) \\
v_{p|n} & = & v_{p|p} + c_p^2( v_{p+1|n} - v_{p+1|p} ),
\end{eqnarray*}
with $c_p = \alpha m_{p|p}/m_{p+1|p}$, where $m_{p+1|p}$ and $v_{p+1|p}$ are the predicted mean and variance at time $p+1$ given the observations $y_0,\dots,y_p$. It follows that the mean $m_p$ and variance $v_p$ of $\pi_{p,0}$ satisfy similar relations to the filtered means and variances:
$$
m_p = \sum_{i = 0}^p m_{i|i} \alpha^i (1 - \mathbb{I}_{i < p}\alpha^2 d_p) \prod_{j=0}^{i-1} d_j
\quad\mbox{ and }\quad
v_p = \sum_{i = 0}^p v_{i|i} \alpha^{2i} (1 - \mathbb{I}_{i < p}\alpha^4 d_p^2) \prod_{j=0}^{i-1} d_j^2,
$$
where $d_p = v_{p|p}/v_{p+1|p}$ and where $\mathbb{I}_c$ is the indicator of condition $c$. The objective is to compute the order of
$$
\Pi_{p,0}^{-1}(u) - \Pi_{p-1,0}^{-1}(u) = m_p - m_{p-1} + \sqrt{2}\erf^{-1}(2u-1) (\sigma_p - \sigma_{p-1})
$$
where $\sigma_p = \sqrt{v_p}$. From the above expression, it follows easily that
$$
m_p - m_{p-1} = \alpha^p(m_{p|p} - m_{p|p-1}) \prod_{i=0}^{p-1} d_i
\quad\mbox{ and }\quad
v_p - v_{p-1} = \alpha^{2p}(v_{p|p} - v_{p|p-1}) \prod_{i=0}^{p-1} d_i^2.
$$
which yields the same order for both $m_p - m_{p-1}$ and $\sigma_p - \sigma_{p-1}$. The desired result follows from the fact that
$$
\alpha^p \prod_{i=0}^{p-1} d_i = \alpha^p \prod_{i=0}^{p-1} \dfrac{v_{i|i}}{\alpha v_{i|i} + \beta^2} = \prod_{i=0}^{p-1} \dfrac{\alpha}{\alpha^2 + \beta^2/v_{i|i}} = \prod_{i=0}^{p-1} \bigg( \alpha + \dfrac{\beta^2}{ \alpha v_{i|i}} \bigg)^{-1},
$$
and from the assumption that $v_{p|p} = \mathbb{V}\textrm{ar}( X_p \mid y_{0:p} ) \approx \gamma^2$ for all $p$ large enough.
\end{proof}


\begin{thebibliography}{9}

\bibitem{andrieu} 
\textsc{Andrieu}, C., \textsc{Doucet}, A. \& \textsc{Holenstein},
R.~(2010). Particle Markov chain Monte Carlo methods (with discussion).
\textit{J. R. Statist. Soc. Ser. B}, \textbf{72}, 269--342.

\bibitem{Cappe_2005}
{\sc Capp\'e}, O., {\sc Ryden}, T, \& {\sc Moulines}, \'E.~(2005). \emph{Inference
in Hidden Markov Models}. Springer: New York.


\bibitem{chopin3}
{\sc Chopin}, N. \& {\sc Singh}, S. S.~(2015). On particle Gibbs sampling. \emph{Bernoulli}, {\bf 21}, 1855-1883.

\bibitem{dds1}
{\sc Del Moral}, P., {\sc Doucet}, A. \& {\sc Singh}, S. S.~(2010). 
A backward interpretation of Feynman-Kac formulae. \emph{M2AN}, {\bf 44},
947--975.

\bibitem{doucet_johan}
{\sc Doucet}, A. \& {\sc Johansen}, A. (2011). A tutorial on particle filtering and smoothing: Fifteen years later. In \emph{
Handbook of Nonlinear Filtering} (eds. D. Crisan \& B. Rozovsky), Oxford University Press: Oxford.

\bibitem{edwards}
{\sc Edwards}, D. A.~(2011). On the Kantorovich-Rubinstein theorem. \emph{Expos. Math.}, {\bf 29}, 387--398.

\bibitem{giles}
{\sc Giles}, M. B.~(2008). Multilevel Monte Carlo path simulation. \emph{Op. Res.}, {\bf 56}, 607-617.

\bibitem{hein}
{\sc Heinrich}, S.~(2001).
{Multilevel {Mo}nte {C}arlo methods}.
In \emph{Large-Scale Scientific Computing}, (eds.~S. Margenov, J. Wasniewski \&
P. Yalamov), Springer: Berlin.

\bibitem{jacob}
{\sc Jacob}, P., {\sc Lindsten}, F. \& {\sc Schon}, T.~(2017). 
Smoothing with Couplings of Conditional Particle Filters.
arXiv preprint.

\bibitem{jasra_smooth}
{\sc Jasra}, A.~(2015). On the behaviour of the backward interpretation of Feynman-Kac formulae under verifiable conditions. \emph{J. Appl. Probab.}, {\bf 52}, 339--359.

\bibitem{mlpf}
{\sc Jasra}, A., {\sc Kamatani}, K., {\sc Law} K. J. H. \& {\sc Zhou}, Y.~(2017). 
Multilevel particle filters. \emph{SIAM J. Numer. Anal.}, {\bf 55}, 3068-3096.

\bibitem{kantas}
 {\sc  Kantas}, N., {\sc Doucet}, A.,  {\sc Singh}, S. S.,  {\sc Maciejowski}, J. M. \& {\sc Chopin}, N.~(2015)
On Particle Methods for Parameter Estimation in General State-Space Models.
\emph{Statist. Sci.}, {\bf 30},  328-351.


\bibitem{Kumar}
{\sc Kumar}, P.R., \& {\sc Pravin} V.\ (1986). \emph{Stochastic systems: Estimation, identification, and adaptive control}. Prentice-Hall.

\bibitem{olsson}
{\sc Olsson}, J. \& {\sc Westerborn}, J.~(2017).
Efficient particle-based online smoothing in general hidden Markov models: The PaRIS algorithm.
\emph{Bernoulli}, {\bf 23}, 1951--1996.

\bibitem{rauch}
{\sc Rauch}, H. E., {\sc Striebel}, C. \& {\sc Tung}, F.~(1965)  Maximum likelihood estimates of linear dynamical systems. \emph{AIAA J.} {\bf 3}, 1445--1450.

\bibitem{rg:15}
{\sc Rhee,} C. H., \& {\sc Glynn}, P. W.~(2015).
Unbiased estimation with square root convergence for SDE models.
\emph{Op. Res.}, {\bf 63}, 1026--1043.

\bibitem{span}
{\sc Spantini}, A., {\sc Bigoni}, D. \& {\sc Marzouk} Y.~(2017). Inference via low-dimensional couplings.
arXiv preprint. 


\end{thebibliography}
\end{document}